\documentclass[pra,aps,twocolumn,10pt]{revtex4-1}
\usepackage[strict]{revquantum}
\usepackage[bold]{hhtensor}

\definecolor{darkblue}{RGB}{0,0,127} 
\definecolor{darkgreen}{RGB}{0,150,0}
\hypersetup{breaklinks, colorlinks, linkcolor=darkblue, citecolor=darkgreen, filecolor=red, urlcolor=darkblue}
\hypersetup{
	pdftitle={Maximum likelihood quantum state tomography is inadmissible}, 
	pdfauthor={Christopher Ferrie and Robin Blume-Kohout}
}

\newcommand{\ketbra}[2]{\ket{#1}\!\!\bra{#2}}

\newcommand{\proj}[1]{\ketbra{#1}{#1}}
\newcommand{\qexpect}[1]{\ensuremath{\left\langle#1\right\rangle}}

\definelanguagealias{EN}{english}

\DeclareMathOperator*{\argmin}{argmin}

\begin{document}

\title{Maximum likelihood quantum state tomography is inadmissible}

\author{Christopher Ferrie}
\email{csferrie@gmail.com}
\homepage{https://www.csferrie.com}

\affiliation{University of Technology Sydney,
Centre for Quantum Software and Information,
Ultimo NSW 2007, Australia}


\author{Robin Blume-Kohout}
\affiliation{Center for Computing Research, Sandia National Laboratories, Albuquerque, New Mexico, 87185}

\date{\today}

\begin{abstract}
Maximum likelihood estimation (MLE) is the most common approach to quantum state tomography.  In this letter, we investigate whether it is also \emph{optimal} in any sense.  We show that MLE is an \emph{inadmissible} estimator for most of the commonly used metrics of accuracy---i.e., some other estimator is more accurate for every true state.  MLE is inadmissible for fidelity, mean squared error (squared Hilbert-Schmidt distance), and relative entropy.  We prove that almost \emph{any} estimator that can report both pure states and mixed states is inadmissible.  This includes MLE, compressed sensing (nuclear-norm regularized) estimators, and constrained least squares.  We provide simple examples to illustrate why reporting pure states is suboptimal even when the true state is itself pure, and why ``hedging'' away from pure states generically improves performance.
\end{abstract}

\maketitle

 
\textbf{Introduction:} Quantum state tomography means reconstructing---or, more precisely, \emph{estimating}---a quantum state $\rho$, using data $D$ that are the results of measurements on many copies of $\rho$.  The question ``When and how can a state be estimated accurately?''~goes back at least to Pauli \cite{pauli_handbuch_1933}, but has been intensively studied since the 1990s because of its applications in quantum information science.  Reliable tomography is now fairly routine, especially in cases where there is some engineered target state in mind.  Tomography scales poorly with system size---it requires resources that scale exponentially with the number of degrees of freedom---but remains a useful and essential ingredient in engineering small quantum devices.  Since about 2000, the most common analysis method for tomographic data has been \emph{maximum likelihood estimation} (MLE) \cite{hradil_quantum-state_1997,JamesPRA2001}. The MLE estimate is the state $\hat{\rho}$ that maximizes the \emph{likelihood function}, which is defined as the probability of the actually-observed data, given $\rho$:
\begin{equation}\label{eq:MLE}
\mathcal{L}(\rho) \equiv \mathrm{Pr}(D|\rho);\ \ \hat\rho = \mathrm{argmax}[\mathcal{L}(\rho)]
\end{equation}
MLE is convenient, ubiquitous, and grounded in a well-studied statistical principle.  Sometimes the simplest solutions are also the best, and since MLE has received so much attention, it is reasonable to wonder whether it is also optimally accurate in some sense.  Should it be preferred over other estimators \emph{in principle}, not just as a convenient tool?  The answer turns out to be ``no''.

\begin{figure}[ht]
	\centering
	\includegraphics[width=1.0\columnwidth]{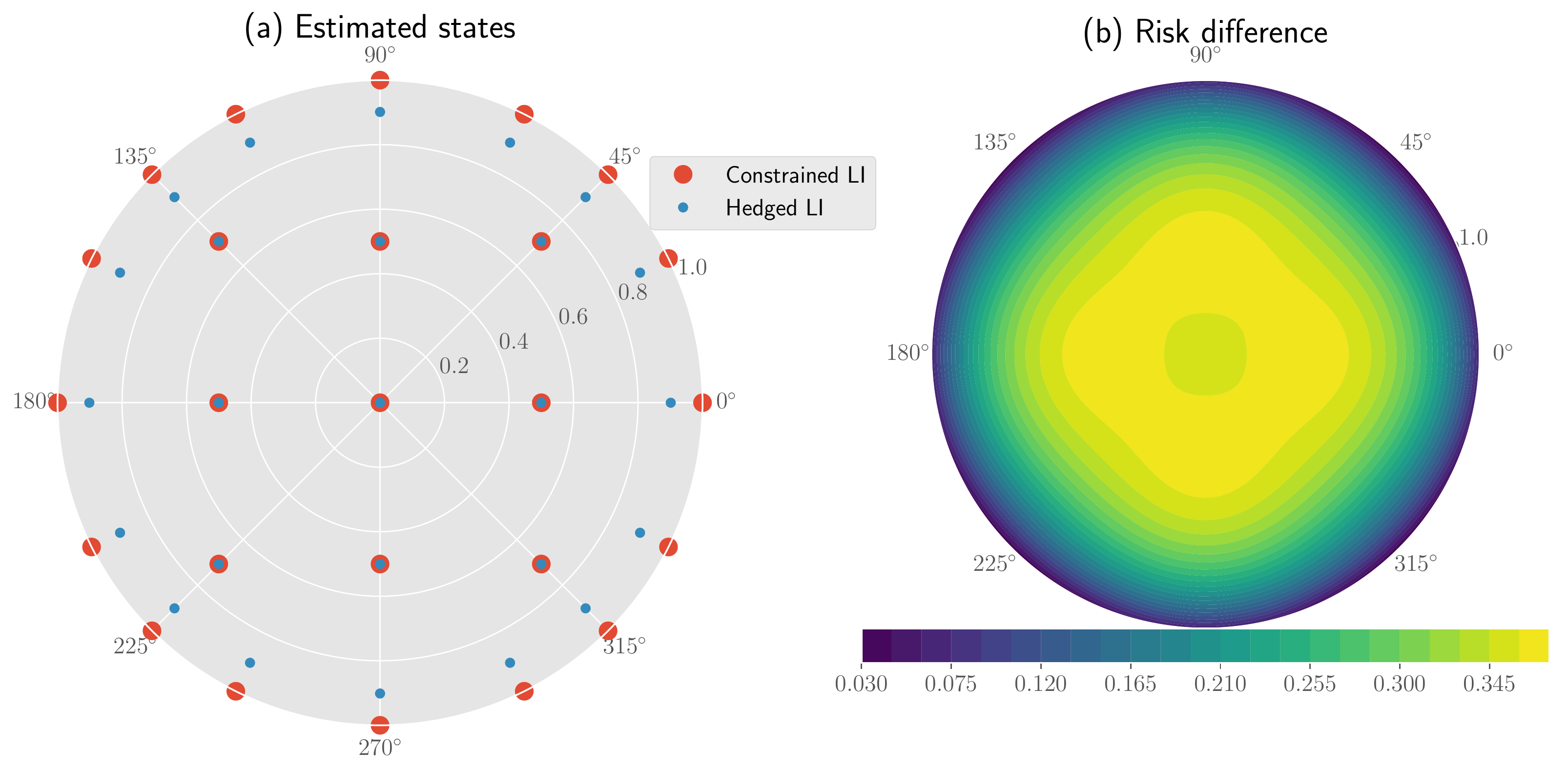}
	\caption{
	\label{fig:rebit}
\textbf{Estimators and Hilbert-Schmidt risk for $N=4$ Pauli measurements on a rebit.} Panel (a) shows the estimates reported by the constrained least-squares estimator (Eq. \ref{eq:li}) and its hedged counterpart (Eq. \ref{eq:hli}). All the hedged estimates lie \emph{inside} the state space boundary.  Panel (b) shows the difference in the Hilbert-Schmidt risk between the constrained least-squares and hedged estimators.  The constrained least-squares estimator is outperformed everywhere, which demonstrates its inadmissibility.}
\end{figure}

The MLE of a quantum state depends, of course, on the observed data $D$.  For some datasets, the estimate can be a pure state.  This seems fairly innocuous.  Other estimators, including nuclear-norm regularization (``compressed sensing'') \cite{gross_quantum_2010,flammia_quantum_2012,RiofrioNC2017} and constrained least-squares (``positivized linear inversion'') \cite{kaznady_numerical_2009,smolin_efficient_2012,qi_quantum_2013}, share the same property.  But reporting pure estimates turns out to be dangerous.  Suppose that $\hat\rho(D)$ is an estimator, and there are two datasets $D_1$ and $D_2$ that are consistent with the same states---meaning that the states $\rho$ for which $\mathrm{Pr}(D_1|\rho)=0$ and $\mathrm{Pr}(D_2|\rho)=0$ are the same---and one of the estimates $\hat\rho(D_1)$ and $\hat\rho(D_2)$ is pure while the other is mixed.  (In most situations, this describes MLE and the other estimators mentioned).  Then, as we will show, this estimator is suboptimal in a very strong sense: it is \emph{inadmissible} as measured by the infidelity error metric.  Admissibility is a standard concept from statistical decision theory \cite{HinckleyCox79}, defined precisely below.  But in essence a inadmissible estimator is \emph{strictly} less accurate than some other estimator.  It may be used for convenience, but is generally excluded from discussions of optimality.

Admissibility depends on the error metric used to define ``accuracy''.  Infidelity (mentioned above) is not the only metric.  Other error metrics used to evaluate quantum tomography include mean-squared error (squared Hilbert-Schmidt distance) and quantum relative entropy. MLE and the other estimators mentioned are also inadmissible according to these metrics.  We will show that if an estimator can report two distinct pure states, at least one of which is not ruled out by the dataset that produced the other, then it is inadmissible for relative entropy and mean-squared error.

Our main goal in the rest of this Letter is to prove these statements.  We then discuss their implications, and give some simple examples to illustrate them.  We begin with some necessary background information.

\noindent\textbf{Background:} Tomography is done by measuring $N$ identically prepared systems and recording the empirical frequencies of the measurement outcomes \footnote{Usually, the systems are split into subsets, and each subset is measured a different way.  But because estimation requires only the observed outcomes (not how they were obtained), those various measurements can be conglomerated into a single POVM $\{E_k\}$ without loss of generality.}.  Measuring a quantum system yields one of several outcomes that we label by $k$, each of which can be represented by a positive operator $E_k$, with the requirement that $\sum_k E_k = \id$.  The outcomes' probabilities depend on the system's state $\rho$ according to Born's rule:
\begin{equation}\label{eq:born}
    \Pr(k|\rho) = \Tr(\rho E_k) := p_k.
\end{equation}
As $N\to\infty$, each empirical (observed) frequency will approach its underlying probability almost certainly:
\begin{equation}\label{eq:freq}
    f_k=\frac{n_k}{N} \to p_k \text{ as } N\to\infty,
\end{equation}
where $n_k$ is the number of times that ``$k$'' was observed.

The oldest and simplest tomographic estimator is \emph{linear inversion}, which simply inverts the system of linear equations implied by Eqs. \autoref{eq:born}-\autoref{eq:freq} to determine $\rho$ (up to finite-$N$ fluctuations).  \emph{Least squares} estimation generalizes linear inversion to the situation where these equations are overdetermined.  Both are subject to a well-known pathology:  the estimate $\hat\rho$ may have negative eigenvalues (and therefore not be a quantum state at all).  This can be patched by truncating $\hat\rho$ and renormalizing \cite{kaznady_numerical_2009}, or more elegantly by \emph{constrained least squares} \cite{smolin_efficient_2012}, which restricts the least-squares optimization to the convex set of valid quantum states defined by $\Tr\rho=1$ and $\rho\geq0$.  These estimators can be seen as simplified approximations to MLE \cite{hradil_quantum-state_1997}, which is motivated by a vast literature in classical statistics and defined by Eq. \autoref{eq:MLE}.  

For each of these estimators, certain datasets produce estimates that are pure states \cite{blume-kohout_optimal_2010}.  This follows from some fairly simple geometry.  For any given true state $\rho$, the empirical frequencies $f_k$ of the data are random variables that fluctuate around (but do not generally coincide with) the underlying probabilities $p_k$.  So the linear inversion estimate is drawn randomly from a roughly spherical distribution whose mean is $\rho$.  All the estimation procedures defined above, including MLE, seek the valid quantum state that is ``closest'' to the linear inversion estimate by some metric.  If $\rho$ is a pure state, then the cloud of possible linear inversion estimates must englobe it, and some of them will be closer to a pure state than any mixed state \cite{ScholtenNJP2018}.

It is worth noting that pure estimates are not exceptional or rare.  For example, if the true state is almost any pure qubit state $\rho=\proj{\psi}$, then as $N\to\infty$, the MLE or constrained least-squares estimate will be pure with probability $1/2$ \cite{ScholtenNJP2018}.  More recently, \emph{compressed sensing} estimators based on the theory of matrix completion have been proposed, which actually prefer low-rank estimates \cite{flammia_quantum_2012}.  So pure estimates are not only common but, in some contexts, even sought-after.  

To investigate whether they are desirable---i.e., whether reporting pure states (sometimes) is a good strategy for improving an estimator's accuracy or performance---we turn to quantum statistical decision theory \cite{helstrom_quantum_1976,holevo_probabilistic_1982}.  ``Inaccuracy'' corresponds to the cost of reporting $\hat\rho$ when the true state is $\rho$, and is quantified by a \emph{loss function} $L(\rho,\hat\rho)$.  Estimators are maps from data to estimates $\hat\rho$.  Each true state defines a distribution over data, $Pr(D|\rho)$.  For each estimator and each true state $\rho$, the expected loss (average over data) is called that estimator's \emph{risk} at $\rho$.
\begin{equation}\label{eq:risk}
    R(\rho,\hat\rho) = \sum_{D}\Pr(D|\rho) L(\rho,\hat\rho(D)).
\end{equation}

There is no universal agreement on what loss function to use in quantum state tomography, but some of the most common choices are:
\begin{enumerate}
\item mean squared error (squared Hilbert-Schmidt distance) $H(\rho,\hat\rho) = \Tr(\rho-\hat\rho)^2$;
\item quantum relative entropy $D(\rho,\hat\rho) = \Tr\rho(\log\rho-\log\hat\rho)$; and,
\item infidelity $F(\rho,\hat\rho) = 1-\Tr\sqrt{\sqrt{\rho}\hat\rho\sqrt{\rho}}$.
\end{enumerate}
Risk depends on the true state $\rho$.  The best estimator would have low risk on \emph{all} states, but it is easy to see that no estimator can be ``best'' for all states.  For any $\rho_0$, the lowest possible risk is achieved by the estimator that always reports $\rho_0$.  Its risk is zero at $\rho=\rho_0$, but catastrophic almost everywhere else.  So there's no ``best everywhere'' estimator, and tradeoffs are necessary.  

Most pairs of estimators are \emph{incomparable}, meaning that each outperforms the other somewhere. But if one estimator is matched or outperformed by another one everywhere, we say that it is \emph{dominated}.  It is hard to recommend such an estimator!  An estimator is \emph{admissible} only if it is not dominated.  An \emph{inadmissible} estimator is strictly worse than some other estimator. 

\noindent\textbf{Main result}:  We will now prove that an estimator of quantum states is inadmissible---for \emph{all} the loss functions given above---if it satisfies the following condition: \emph{for data sets that are inconsistent with the same set of quantum states, the estimator can report at least one mixed state and also can report at least two pure states, one of which is not ruled out by the other's dataset}.  This condition deserves a bit of explanation.  Essentially, there are a couple of experimental designs for which MLE (and similar estimators) are not inadmissible.  One is the ``classical'' case where all the $\{E_k\}$ commute.  Another is the extreme small-sample limit, where (for example) the data comprise a single measurement of each single-qubit Pauli basis.  This condition (made precise in Theorems 1-2) concisely excludes all the relevant edge cases, but is generally satisfied by MLE (more on this later).

The proof makes use of Wald's complete class theorem \cite{wald_essentially_1947}.  It states that if an estimator $\hat\rho$ is admissible, then under weak regularity conditions, there exists a distribution $\Pr(\rho)d\rho$ such that
\begin{equation}
    \hat\rho = \argmin_{\hat\sigma} \int R(\rho,\hat\sigma)\Pr(\rho)d\rho.
\end{equation}
For this distribution, $\hat\rho$ minimizes the posterior \emph{Bayes risk},
\begin{equation}
    B(\hat\rho) = \int \Pr(\rho|D)L(\rho,\hat\rho(D))d\rho,
\end{equation}
and is called a \emph{Bayes estimator} for that posterior.  Wald's theorem establishes a duality between frequentist and Bayesian optimality, by saying that every admissible estimator is also a Bayes estimator that minimizes \emph{average} risk for some prior---and thus that if $\hat\rho$ is not a Bayes estimator for \emph{any} distribution, it cannot be admissible.  

The minimal assumptions on the estimator are slightly different for fidelity than for relative entropy and mean squared error, so we prove the general result in two steps.

\begin{theorem}
Let $\hat\rho$ be an estimator of quantum states. Suppose there exist two datasets $D_1,D_2$ for which $\sigma_1 = \hat\rho(D_1)$ and $\sigma_2 = \hat\rho(D_2)$ are both pure states, and at least one of $Pr(D_1|\sigma_2)$ or $Pr(D_2|\sigma_1)$ is nonzero.  Then $\hat\rho$ is inadmissible for any Bregman divergence.
\end{theorem}
\noindent\emph{Remark}: \emph{Bregman divergences} are a class of risk functions derived from convex ``entropy'' functions, including mean squared error and quantum relative entropy.  The Bayes estimator for any Bregman divergence is the mean of the posterior distribution \cite{blume-kohout_accurate_2006}.

\begin{proof}
Suppose $\hat\rho$ \emph{is} admissible for a Bregman divergence.  Then each of the $\sigma_i$ must be the mean of some posterior distribution that is given by the product of a likelihood and a prior, $Pr(\rho|D_i) = Pr(D_i|\rho)Pr(\rho)$.  But a distribution whose mean is a pure state can only be a point measure \emph{on} that pure state, so $Pr(\rho|D_i) = \delta(\rho-\sigma_i)$.  So the prior $Pr(\rho)$ must have support on both $\sigma_1$ and $\sigma_2$, \emph{and} each of the the likelihoods $Pr(D_i|\rho)$ must rule out the other state, i.e. $Pr(D_1|\sigma_2) = Pr(D_2|\sigma_1) = 0$.  This contradicts the assumption of the Theorem, so $\hat\rho$ is not admissible for any Bregman divergence.
\end{proof}

Our second theorem requires two preliminary lemmas.
\begin{lemma}
If a prior has support only on pure states, then the Bayes estimator for infidelity risk is always pure (no matter what data were observed).
\end{lemma}
\begin{proof}
This has been proven multiple times in the literature (e.g. \cite{blume-kohout_accurate_2006,kueng_near-optimal_2015}).  It follows from the fact that $F(\rho,\sigma)$ is bilinear when either argument is pure.
\end{proof}
\begin{lemma}
If a posterior distribution $\pi$ has support on at least one mixed (non-pure) state, then the Bayes estimator is never pure.
\end{lemma}
\begin{proof}
Suppose $\hat\rho$ is a pure state.  Let $\mathcal{H}$ be the support of a mixed state upon which $\pi$ is supported, $d$ its dimension, and $\mathbb I$ the projector onto $\mathcal{H}$.  If we vary $\hat\rho$ slightly by extending it onto $\mathcal{H}$ to produce $\hat\rho' = \hat\rho + x (\mathbb I/d - \hat\rho)$, then the change in the Bayes reward is (with $\langle\cdot\rangle$ denoting the average over the posterior)
\begin{equation}
	B(\hat\rho') - B(\hat\rho) = \frac{x}2\left(\frac1d \Tr\left\langle\rho(\sqrt{\rho}\hat\rho\sqrt{\rho})^{-\frac12}\right\rangle - B(\hat\rho)\right).
\end{equation}
The second term [$-B(\hat\rho)$] is negative, but finite and bounded.  The first term, however, diverges to $+\infty$ when $\rho$ has support on some subspace that $\hat\rho$ does not.  So $\hat\rho'$ has a higher Bayes reward than $\hat\rho$, $\hat\rho$ is not the Bayes estimator, and therefore the Bayes estimator is not pure.
\end{proof}

\begin{theorem}
Let $\hat\rho$ be an estimator of quantum states.  Suppose there exist two datasets $D_0,D_1$ that are inconsistent with exactly the same set of states---i.e., for all $\rho$, $\Pr(D_0|\rho) = 0 \Leftrightarrow Pr(D_1|\rho)=0$---and $\hat\rho(D_0)$ is mixed and $\hat\rho(D_1)$ is pure.  Then $\hat\rho$ is inadmissible for infidelity loss.
\end{theorem}

\begin{proof}
The proof is by contradiction.  Assume $\hat\rho$ \emph{is} admissible.  Then there is a prior for which it is Bayes.  The posterior distributions corresponding to $D_0$ and $D_1$, which we denote ${\rm Post}(D_0)$ and ${\rm Post}(D_1)$, are different.  But because $D_0$ and $D_1$ were assumed to rule out the same states, ${\rm Post}(D_0)$ and ${\rm Post}(D_1)$ must have the same support.  If that support were restricted to pure states, then by Lemma 1, both $\hat\rho(D_0)$ and $\hat\rho(D_1)$ would have to be pure.  But if that support included a mixed state, then by Lemma 2, neither $\hat\rho(D_0)$ and $\hat\rho(D_1)$ could be pure.  Since the theorem assumed one is pure and one is mixed, we have a contradiction, and $\hat\rho$ is not admissible.
\end{proof}

\noindent\emph{Comment:}  The technical conditions in these theorems---e.g., that either $Pr(D_1|\sigma_2)$ or $Pr(D_2|\sigma_1)$ is nonzero (in Theorem 1), or that $D_1$ and $D_2$ are inconsistent with the same states (in Theorem 2)---are liable to obscure the main point.  These are \emph{not} very demanding conditions!  They hold in the vast majority of cases, and as noted above they are only necessary to exclude some obscure and uncommon experimental designs.  These conditions will hold whenever (1) each measurement performed has finitely many outcomes, and (2) each measurement is repeated $N\gg1$ times.  In that ubiquitous situation, it is sufficient to consider datasets in which each possible measurement outcome has occurred at least once, for which the technical conditions are guaranteed to hold.

\noindent{\bf Examples:} Theorems 1 and 2 are abstract, and provide relatively little intuition.  To illustrate how and why estimators like MLE are inadmissible, we now consider some simple single-qubit examples.  In these examples, we compare two estimators using Hilbert-Schmidt risk.  The first is constrained least-squares, which reports pure states and satisfies the conditions of both theorems and is therefore inadmissible for all the metrics discussed.  We demonstrate its inadmissibility explicitly by constructing a ``hedged'' estimator that outperforms it everywhere.

For our first example, we assume that the qubit's state lies in the X-Z plane.  This makes it a \emph{rebit}, and allows easy visualization of the 2-dimensional ``Bloch disk'' state space.  We represent each state $\rho$ by a Bloch vector $\vec{r} = (\qexpect{X},\qexpect{Z})$.  Since $\qexpect{Y}=0$ is known, we let $N$ samples of the unknown state be measured in each of the Pauli $X$ and $Z$ bases, yielding $n_X$ and $n_Z$ ``+1'' outcomes, and thus the empirical frequencies are
\begin{equation}
	\vec{f} = \left(\frac{2n_X-N}N,\frac{2n_Z-N}N\right).
\end{equation}
The constrained least-squares estimator is given by
\begin{equation}
	\vec{r} = \begin{cases}
	\vec{f} & \text{ if }  \|\vec{f}\|^2 < 1 \\
	\vec{f}/\|\vec{f}\| & \text{ otherwise} 
	\end{cases}, \label{eq:li}
\end{equation}
and is illustrated in Figure \ref{fig:rebit}(a).  For $N>1$, it satisfies the conditions of the theorems and is therefore (like MLE) inadmissible.  

To demonstrate this explicitly, we construct a strictly better estimator by \emph{hedging} \cite{blume-kohout_hedged_2010}.  Hedging describes any modification to the estimator that forces estimates to be full-rank (i.e., inside the state space rather than on its boundary).  Our ad-hoc hedged estimator is
\begin{equation}
	\vec{r}_h = \begin{cases}
	\vec{f} & \text{ if }  \|\vec{f}\|^2 < 1\\
	\sqrt{1-h}\vec{f}/\|\vec{f}\| & \text{ otherwise} 
	\end{cases}. \label{eq:hli}
\end{equation}
Unlike the general hedging procedure defined in \cite{blume-kohout_hedged_2010}, this simpler procedure \emph{only} affects pure estimates, but is still sufficient to outperform the unhedged estimator.  Its performance depends on the strength of the hedging parameter ($h$).  By evaluating the estimator's expected Hilbert-Schmidt risk at $\vec{r} = (0,1)$ and maximizing it (a straightforward calculation), we identify a good value of $h$,
\begin{equation} \label{eq:hopt}
	h=  \frac{1}{N}-\frac{1}{N^2}.
\end{equation}
The estimator defined by Eqs. \autoref{eq:hli} and \autoref{eq:hopt} is also shown in Fig. \ref{fig:rebit}(a).  We computed the Hilbert-Schmidt risk of both estimators for $N=4$ measurements of each observable. In Fig. \ref{fig:rebit}(b), we plot the difference, showing that the hedged estimator has lower risk \emph{everywhere}.

The rebit with just $N=4$ measurements is easy to calculate and visualize, but not very realistic.  So for our second example, we consider a full qubit with measurements of all three Pauli operators and a much wider range of $N$.  The 2-dimensional nature of this paper precludes visualizing the risk for all states in the Bloch ball. Happily, due to the symmetry of the Pauli measurements, it is sufficient to plot the risk difference along (1) a measurement axis ($Z$), and (2) one of the ``magic'' axes that point toward the corners of the Bloch cube. 

Since Hilbert-Schmidt risk scales as $1/N$ for every reasonable estimator, we consider the scaled risk difference,
\begin{equation}
N[R(\vec{r},\vec{r}_0)-R(\vec{r},\vec{r}_h)],
\end{equation}
where $\vec{r}$ is the true state, $\vec{r}_0$ is the constrained least-squares estimate, and $\vec{r}_h$ is the hedged estimate.  Figure \ref{fig:qubit_coin} shows the risk difference along the $Z$ axis for a wide range of $N$, while Fig. \ref{fig:qubit_magic} shows it for the axis pointing toward the point $(1,1,1)$.  In all cases, the hedged estimator dominates constrained least-squares, illustrating why (and how) the constrained least-squares estimator's habit of reporting pure states makes it inadmissible.

\begin{figure}[ht]
\centering
\includegraphics[width=1.0\columnwidth]{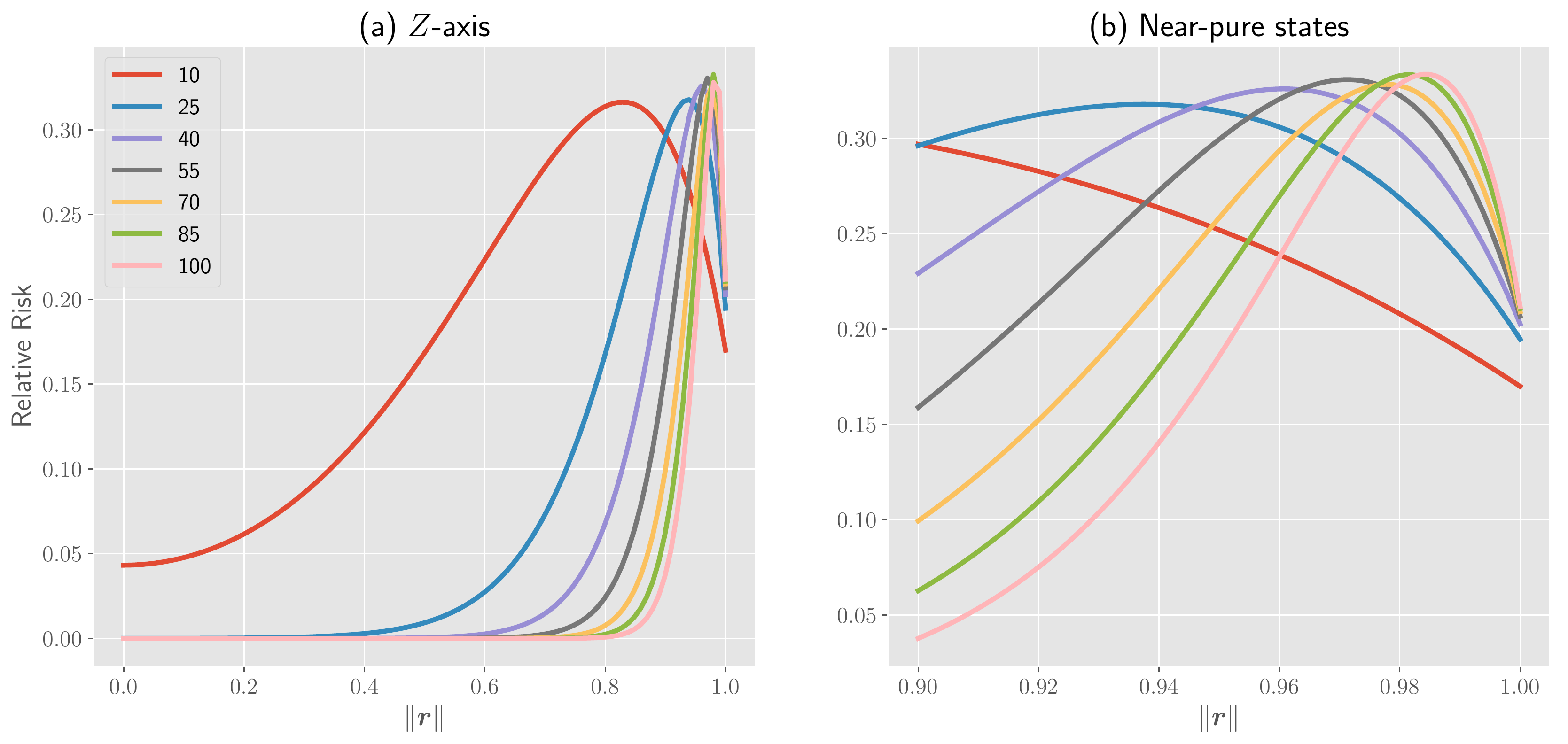}
\caption{
\label{fig:qubit_coin}
\textbf{The hedged estimator outperforms constrained least-squares along the $Z$ axis}.  The relative difference in Hilbert-Schmidt risk between the constrained least-squares and hedging estimators is plotted, for $N=10\ldots 100$ Pauli measurements on a single qubit.  Panel (a) shows the relative risk along the $Z$ axis, as a function of the radial coordinate $r$.  Panel (b) shows the same information, but zoomed in to the range $r=0.9\ldots 1$.  Importantly, the difference is positive everywhere, showing that the hedged estimator dominates its unhedged counterpart. }
\end{figure}

\begin{figure}[ht]
\centering
\includegraphics[width=1.0\columnwidth]{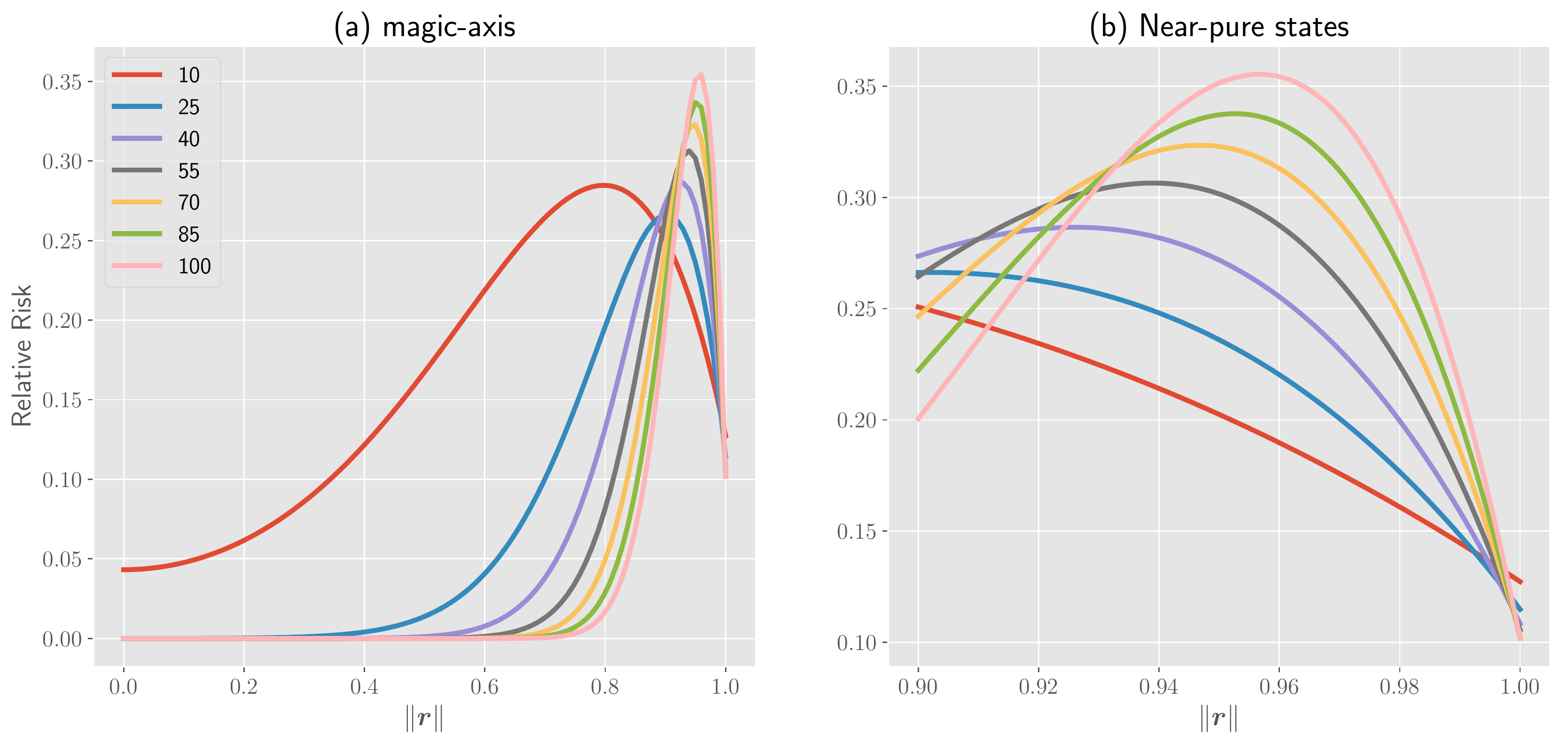}
\caption{
\label{fig:qubit_magic}
\textbf{The hedged estimator outperforms constrained least-squares along the $(1,1,1)$ axis}.  This figure shows exactly the same information as Fig. \ref{fig:qubit_coin}, but along an axis pointing to $\vec{r}=(1,1,1)$ rather than along the Z axis.}
\end{figure}

\noindent\textbf{Discussion:}  The point of this paper is simple:  neither MLE nor any other estimator that reports pure states can be optimally accurate except under very special and restrictive circumstances.  Those exceptional cases \emph{do} exist, although they are rare.  For example, MLE \emph{is} admissible when all the measurements commute.  It's hard to call this ``tomography'', though---it can't identify the quantum state uniquely, and reduces to classical probability estimation.  MLE can also be admissible when each distinct measurement is performed on just one sample.  And there are admissible estimators (not MLE) that \emph{always} report pure states, such as the Bayes estimator for fidelity risk and a prior supported only on pure states.  We leave open the question of whether a similar result holds for trace distance, which is probably the most popular metric not addressed by our theorems.

All of the numerical results shown here can be reproduced by substituting MLE for constrained least-squares.  However, there is no closed-form solution for MLE even on a single qubit, which precludes analytic understanding of those results.  We present the simpler case because it's easier to analyze and understand.

Finally, we want to emphasize that this does \emph{not} mean that MLE (or the other estimators like it) should be avoided.  MLE is reliable, easy, and statistically well-motivated.  It just isn't \emph{optimal} in any reasonable sense.  Its virtues are convenience and simplicity, \emph{not} any claim to uniqueness.  MLE is probably ``accurate enough'' for most purposes---but many other estimators may also be ``accurate enough'', and should not necessarily be thrown out of consideration in favor of MLE.


\begin{acknowledgments}
We thank Joshua Combes, Jonathan Gross, and Marco Tomamichel for helpful discussions. CF was supported by the Australian Research Council Grant No. DE170100421. Sandia National Laboratories is a multimission laboratory managed and operated by National Technology and Engineering Solutions of Sandia, LLC, a wholly owned subsidiary of Honeywell International, Inc., for the U.S. Department of Energy's National Nuclear Security Administration under contract DE-NA0003525.  The views expressed in the article do not necessarily represent the views of the U.S. Department of Energy or the United States Government.
\end{acknowledgments}

\bibliography{tomo}

\end{document}